\documentclass{article}
\usepackage{setspace}
\usepackage[top=1in, bottom=1.25in, left=1.25in, right=1.25in]{geometry}
\usepackage{mathtools}
\usepackage[toc,page]{appendix}
\usepackage{enumerate}
\usepackage{amsmath}
\usepackage{empheq}
\usepackage{amssymb}
\usepackage{amsthm}
\usepackage{listings} 
\usepackage{xcolor} 
\usepackage{graphicx} 

\usepackage{caption}
\usepackage{subcaption}
\usepackage{enumerate}
\usepackage{color}
\usepackage{float}
\usepackage{csvsimple}
\usepackage{comment}
\usepackage[hidelinks]{hyperref}
\usepackage{cleveref}
\usepackage{framed}
\usepackage{mwe}
\usepackage{placeins} 

\newtheorem{theorem}{Theorem}[section]

\theoremstyle{remark}
\newtheorem{remark}{Remark}

\begin{document}

\title{Analysis of a platooned car-following model with different inter-platoon communication levels}

\author{Shouwei Hui \thanks{Shouwei Hui is with the Department
of Mathematics, Univercity of California Davis, Davis,
CA, 95616 USA e-mail: (huihui@ucdavis.edu).}, Michael Zhang \thanks{Michael Zhang is with the Department
of Civil and Environmental Engineering, Univercity of California Davis, Davis,
CA, 95616 USA e-mail: (hmzhang@ucdavis.edu)} \thanks{Corresponding author.}}
\date{}

\maketitle

\begin{abstract}
Despite growing interest in vehicle platooning research, the effect of communication capability between platoons is not investigated to a depth of depth. In this paper, we extend a single-platoon car-following (CF) model to multi-platoon CF models for connected and autonomous vehicles (CAVs) with different inter-platoon communication capabilities. Specifically, we consider forward and backward connection availabilities with delays between platoons. Using linear stability analysis, we discovered that for identical platoons, stability increases with platoon size and connection availabilities and decreases exponentially with large delay. With maximum acceleration and emergency braking constraints integrated into the models, we performed simulations for various cases of CAV platoons and mixed autonomy with human-driven vehicles (HDVs). The simulation results for CAV platoons are consistent with theoretical analysis. The mixed autonomy experiments demonstrate that in the ring road scenario, CAV platoons can stabilize HDVs without adaptions, and the effect of distribution is marginal. Overall, this paper provides valuable insights for designing vehicle-to-vehicle (V2V) communications and managing mixed traffic scenarios.
\\\\
Keywords: CAV platoon, communication capability, linear stability, numerical simulation
\end{abstract}

\section{Introduction}

Platooning is a coordinated driving strategy that keeps a short distance between vehicles, which can enhance capacity, safety and flow rate of traffic networks. It can be conceptualized as a longitudinal traffic control system \cite{sheikholeslam1990longitudinal}. With the recent advancements in connected and autonomous vehicles (CAVs) and associated communication technologies, the era of mixed autonomy, where CAV platoons and human-driven vehicles (HDVs) share the road, is fast approaching. To model such mixed-autonomy environments, one common choice is microscopic traffic flow models, e.g. car-following (CF) models, which provide a useful framework for researchers to study the interactions and dynamics between CAVs and HDVs under varying traffic conditions.

% Car following models : basic
Car-following (CF) models have been fundamental to traffic flow theory since Pipes introduced the first operational model in 1953 \cite{pipes1953operational}. Following this, several CF models were developed to incorporate more realistic dynamics. Notably, the Optimal Velocity Model (OVM) proposed by Bando et al. \cite{bando1995dynamical, bando1995phenomenological} replaced the leading vehicle's velocity in Pipes' model with an optimal velocity function. The Intelligent Driver Model (IDM), introduced by Treiber et al. \cite{treiber2000congested}, further improved upon these models by considering the velocity difference between the lead and following vehicles. These early-stage CF models are widely used in traffic simulations and control design, since they are capable of describing typical traffic phenomena with relatively simple forms. Several extensions have been proposed for these models, e.g.  \cite{jiang2001full}, \cite{yu2013full}, and \cite{derbel2013modified}. However, these models limits the interaction between a pair of vehicles in a leader-follower configuration. Extensions are necessary for these models to describe more complicated traffic scenarios involving CAVs.

% Advanced CF models with controller and CAVs
The next stage of CF models includes multi-vehicle interaction, moving beyond the classic models considering only the immediate leader, thus modeling traffic flow more realistically.CF patterns such as multi-following \cite{lenz1999multi} and backward-following \cite{nakayama2001effect}, both based on OVM, have been increasingly explored and have inspired subsequent research on mixed traffic flow \cite{zhu2018analysis}, \cite{wang2021leading}, where CAVs are designed to better stabilize traffic. In fully autonomous traffic scenarios, various platooning designs can be integrated to enhance both efficiency and safety \cite{jia2016platoon, sun2020relationship, zhang2021internet, zhou2023autonomous, zhou2024self, zong2025platoon}. Further improvements to CF models include the addition of delay factors \cite{davis2003modifications}, \cite{treiber2006delays}, as well as the integration of control mechanisms such as safety-prioritized control \cite{zhao2023safety, zhao2024leveraging} and feedback control \cite{jin2020dynamical}. These extended CF models, coupled with advanced control strategies, reflect the growing trend toward improving traffic stability, efficiency, and safety. By accounting for multi-vehicle interactions and integrating autonomous systems, these models contribute significantly to enhancing traffic performance, particularly in mixed autonomy scenarios.

% Multi-platoon control strategies: ACC, CACC & vehicular communicationsmore 
Control strategies for groups of CAVs, particularly Adaptive Cruise Control (ACC) and Cooperative Adaptive Cruise Control (CACC), have been extensively studied to improve vehicle platooning efficiency and safety \cite{vahidi2003research, milanes2013cooperative}. ACC primarily focuses on maintaining safe distances between vehicles by adjusting speed based on sensor data, achieving better results than typical human drivers. However, it operates in a decentralized manner without relying on vehicle-to-vehicle (V2V) communication. On the other hand, CACC utilizes V2V communication to enable more precise control and coordination among CAV platoons \cite{zheng2015stability, sabuau2016optimal, besselink2017string}. Additionally, the quality of communications plays a crucial role in the stability of CAV platoons, as shown in studies on robust communication and stability analysis \cite{gao2016robust, yu2023assessment, yu2023stability, wang2024stability}. These studies highlight the advantages of CACC over ACC, particularly regarding communication and coordination within platoons.

% Other controls and experiments (MPC, Stochastic, ML, etc).
Beyond ACC and CACC, advanced control strategies such as Model Predictive Control (MPC), reinforcement learning (RL), and stochastic optimization have been explored to further enhance the performance of CAV platoons in complex traffic scenarios. MPC has been widely used due to its ability to predict future states and optimize control actions over a finite horizon \cite{zhou2019distributed, graffione2020model, li2022variable, li2024improved}. Meanwhile, RL has increasingly gained attention for its ability to adaptively learn control policies from data, allowing for more flexible and autonomous decision-making in dynamic environments \cite{liu2022autonomous}. Stochastic optimization approaches have also been applied to account for uncertainties in traffic behavior and communication, as seen in the work of Li \cite{li2017stochastic, li2017stochastic2}. 

Field experiments can further demonstrate the effectiveness of CAVs and platooning in increasing traffic stability and reducing fuel consumption. Various configurations, such as a single CAV leading multiple HDVs on a ring road \cite{stern2018dissipation}, three trucks on a test track \cite{tsugawa2011automated}, and 100 CAVs on a freeway network \cite{lee2024traffic}, have been explored by researchers to highlight the benefits of integrating CAVs and platooning into traffic systems. The data generated from such field experiments can also provide valuable insights for theoretical research, as demonstrated in studies such as \cite{zheng2023comparison, zhou2023experimental}. These studies underscore the potential of advanced control techniques to further enhance the performance and adaptability of CAVs and platooning, especially in mixed autonomy settings.

While car-following and platooning models have been widely studied, most existing studies focus on a single platoon or assume fixed connectivity patterns between platoons, with limited investigation into multi-platoon stability under different inter-platoon connectivity structures. Furthermore, studies on mixed traffic flow involving CAVs and HDVs typically rely on specific control strategies applied to CAVs to improve stability, which often requires processing information from multiple surrounding vehicles. However, the inherent stabilizing effects of structured platooning without complex control mechanisms remain underinvestigated.   

To address these research gaps, this paper proposes a multi-platoon framework that generalizes car-following models to different inter-platoon connectivity structures. We extend a recently proposed single-platoon CF model \cite{hui2024new} to a multi-platoon framework, incorporating varying connection structures and inter-platoon communication delays. Notably, when the platoon size is set uniformly to one, the model degenerates to a classic CF model applicable for HDVs or AVs with backward detection. Theoretically, we proved that the stability of the multi-platoon models depends on the platoon sizes and communication capabilities: larger platoon sizes and enhanced connectivity contribute to increased stability, whereas the negative impact of delays between platoons become significant when they are sufficiently large. Furthermore, we conduct numerical simulations on a ring road, testing different vehicle arrangements, including various CAV platoon sizes and mixed traffic scenarios with different CAV-to-HDV ratios and distributions. The results validate our theoretical findings, and demonstrate that HDVs benefit from following CAV platoons, even when the CAV platoons are not specially designed to control the HDVs.

The remaining part of this paper is organized as follows. In section \ref{sec22}, we introduce the CF models for single and multiple platoons of CAVs. In section \ref{sec23}, the stability criteria of the proposed models are presented and proved. In section \ref{sec24}, we perform numerical simulations for the proposed models with various traffic assignments on a ring road. The impact of delay and connectivity are analyzed. Lastly in section \ref{sec25}, conclusion and possible extensions are given.

\section{Models for CAV platoons \label{sec22}}

\subsection{General assumptions}
We assume that there are $m$ CAV platoons on a single lane road with no overtaking allowed, where $m\geq 1$. The $m$-th platoon is the leading platoon, and $N_i$ denotes the size of the $i$-th platoon. Within the $i$-th platoon the $N_i$-th car is the leading vehicle. Moreover, if the single lane road is a ring road, the $m$-th platoon is following the $1$st platoon.

We select the commonly used Optimal Velocity Model (OVM) \cite{bando1995dynamical} as the base CF model for human driven vehicles (HDVs). The OVM is expressed as
\begin{equation}\label{ch3_OVM1}
    \ddot{x}_i(t)=a\left[V(x_{i+1}(t)-x_{i}(t))-\dot{x}_i(t)\right],
\end{equation}
where $x_i(t)$, $i=1,2,\dots N$ is the position of $i$-th vehicle at time $t$. $x_{i+1}(t)-x_{i}(t)\triangleq h_i(t)$ represents the spatial headway between the $i$-th and $i+1$-th vehicle. $\dot{x}_i(t)$, $\ddot{x}_i(t)$ denotes the velocity and acceleration of the $i$-th vehicle at time $t$, respectively. $V(h)$ is the optimal velocity function of headway (head to head distance) $h$, and $a$ is a sensitivity constant. An example of optimal velocity function is as follows:
\begin{equation}\label{ch3_ovf1}
    V(h)=\begin{cases}
        v_{f}, & \text{if } h\geq h_{f}; \\
        \frac{v_{f}}{2}\left(1-\cos\left(\pi\frac{h-h_{s}}{h_{f}-h_{s}}\right)\right), & \text{if } h_{s}\leq h \leq h_{f};\\
        0, & \text{if } h \leq h_{s},
    \end{cases}
\end{equation}
where $h_{s}$ is the standstill headway, $h_{f}$ is the free flow headway, $v_{f}$ is the free flow speed and $l$ is the length of each vehicle. This is equivalent to the function in \cite{wang2021leading}. Figure \ref{ch3_f1} is an example plot of (\ref{ch3_ovf1}) and the corresponding fundamental diagram (density-flow diagram) as in \cite{hui2024new}.
\begin{figure}[ht]
    \centering
    \subfloat[Optimal velocity function]{\includegraphics[width=0.35\textwidth]{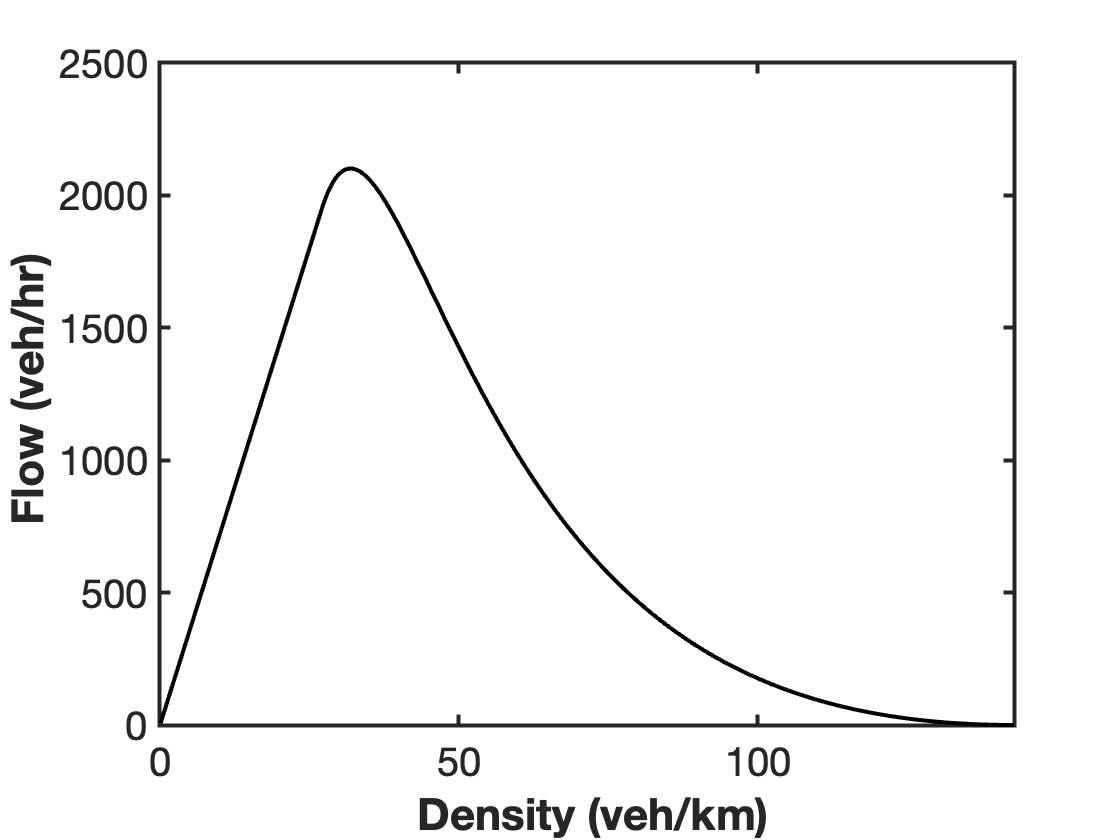}}\label{f1-1}
    \subfloat[Fundamental diagram]{\includegraphics[width=0.35\textwidth]{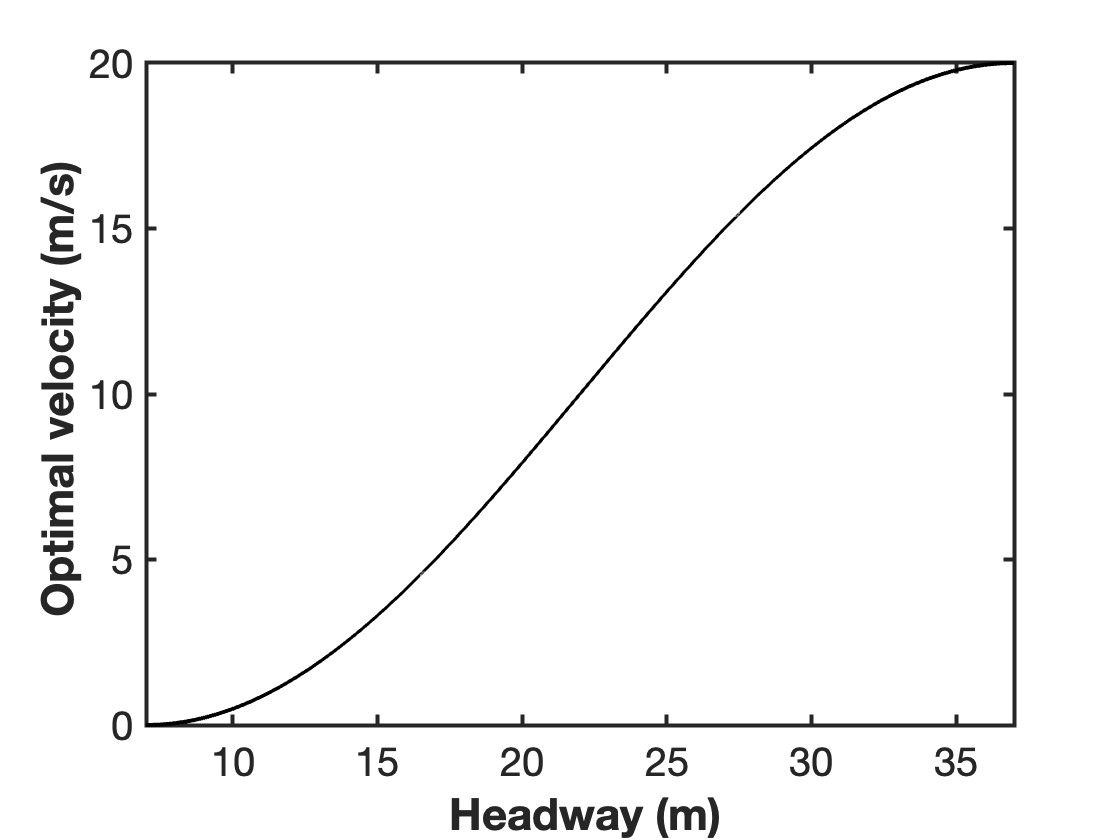}}\label{f1-2}
    \caption{Plot of an optimal velocity function and the corresponding fundamental diagram.}
    \label{ch3_f1}
\end{figure}

\subsection{Single platoon: base model}
Before investigating multi-platoon models, it is essential to develop a robust single-platoon CF model. In \cite{hui2024new}, it is shown that if a platoon is sufficiently close to its equilibrium state, the platoon controlled OVM (P-OVM) is always stable under small initial disturbances and periodic disturbances. The proposed model is of the form
\begin{equation}
    \ddot{x}_i(t)=a\left[V(\frac{x_{N}(t)-x_{i}(t)}{N-i})-\dot{x}_i(t)\right],\;i=1,2,\dots, N\label{ch3_povm}
\end{equation}
where $N$ represents the platoon size, and $x_N$ is the position of the controlled leading vehicle. Figure \ref{fig:onepla} is a visual interpretation of this model. 
\begin{figure}[ht]
    \centering
    \includegraphics[width=0.9\linewidth]{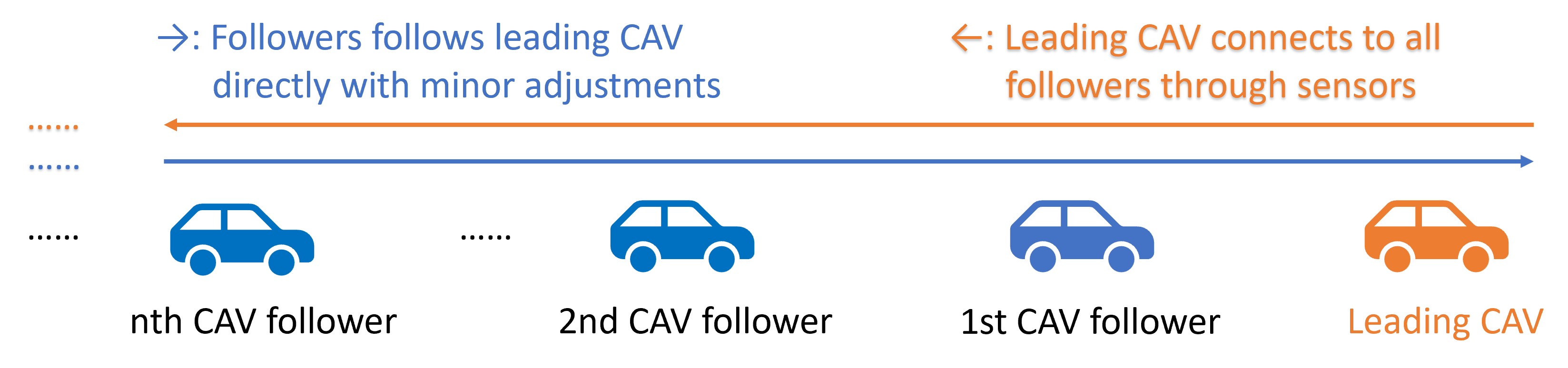}
    \caption{CF pattern of a single platoon}
    \label{fig:onepla}
\end{figure}

However, the reliability of \eqref{ch3_povm} is based on the assumption that platoon followers can precisely acquire information from the leading vehicle without delay, which is unrealistic for excessive amount of CAVs. Therefore, a multi-platoon system (where one platoon follows another) serves as a viable formation strategy for managing long strings of vehicles, helping to mitigate the effects of communication delays and instabilities in centralized controls. In the following subsections, we provide some examples of multi-platoon CF models where each platoon is internally robust and governed by \eqref{ch3_povm}. We will neglect the effects of communication delays and minor adjustments within each platoon.

\subsection{Multi-platoon: no inter-connection \label{subsec23}}

If there is no communication between platoons, we assume that the leading vehicle of each platoon only follows the last vehicle of the platoon ahead, as shown in Figure \ref{multiplafig1}. 
\begin{figure}[ht]
    \centering
    \includegraphics[width=0.9\linewidth]{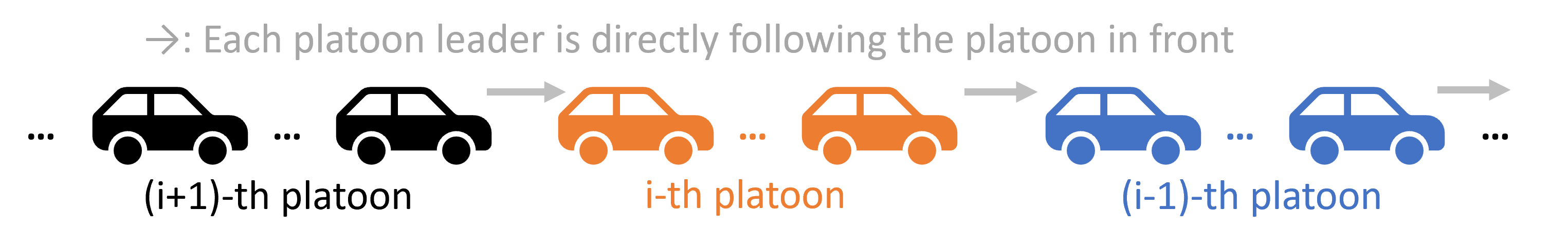}
    \caption{CF pattern of multiple platoons with no inter-platoon communication.}
    \label{multiplafig1}
\end{figure}
If no additional control is applied, the platoon leaders are directly following the vehicle ahead according to the OVM:
\begin{equation}
    \ddot{x}_{i,N_i}=a(V(x_{i+1,1}-x_{i,N_i})-\dot{x}_{i,N_i}),
    \label{leadernocon}
\end{equation}
where $1\leq i \leq m-1$. For the followers within each platoon, the centralized platoon controller \eqref{ch3_povm} is applied:
\begin{equation}
    \ddot{x}_{i,j}=a\left[V\left(\frac{x_{i,N_i}-x_{i,j}}{N_i-j}\right)-\dot{x}_{i,j}\right],
    \label{follower}
\end{equation}
where $1\leq i \leq m$ and $1\leq j \leq N_i-1$. In particular, if we have platoon size $N_m=1 \text{ or } 2$ for all $m$, the proposed model reduces to the OVM for HDVs. 

\subsection{Multi-platoon: two-way inter-connection \label{subsec24}}

In this subsection we suppose each platoon can communicate with platoons ahead and behind, with some delays (i.e. the platoons are two-way connected),  as shown in Figure \ref{multiplafig2}.
\begin{figure}[ht]
    \centering
    \includegraphics[width=0.9\linewidth]{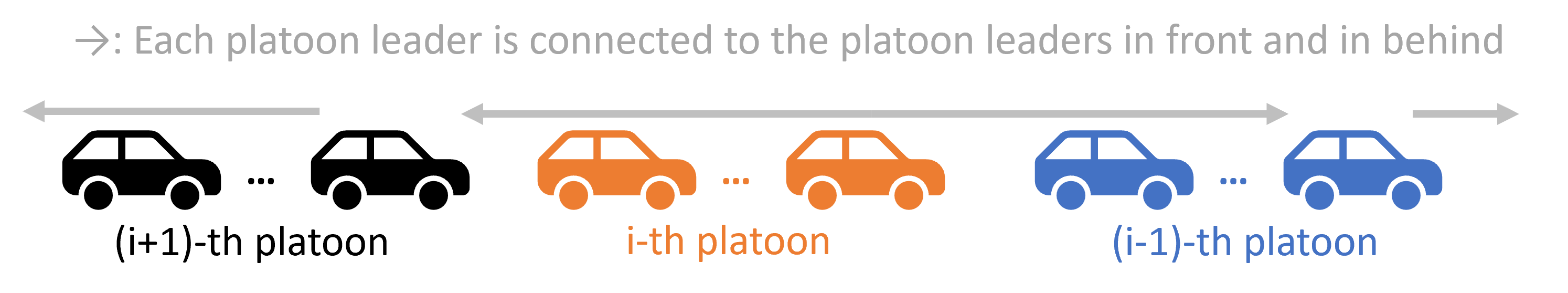}
    \caption{CF pattern of multiple platoons with forward and backward inter-platoon communication.}
    \label{multiplafig2}
\end{figure}
In this case, we can model the platoon leaders to follow both the platoon leader in front and the one behind, similar to the model for autonomous vehicles proposed in \cite{zhu2018analysis}:
\begin{equation}
\begin{aligned}
    \ddot{x}_{i,1}=a&\left[ (1+p)V\left(\frac{\Delta_{i-1}(t-t_{df})}{N_{i-1}}\right)-pV\left(\frac{\Delta_{i}(t-t_{db})}{N_{i+1}}\right)-\dot{x}_{i,1}\right],
\end{aligned}
\label{mp-ovm}
\end{equation}
where $\Delta_i(t)=x_{i+1,1}(t)-x_{i,1}(t)$ is the headway between the $i$-th and $i+1$-th platoon leader, $p$ is the smoothing factor for the platoon in the back, and $t_{df},\;t_{db}$ are the forward and backward communication delays (which can be non-constant) between platoons. For the followers we apply the same equation \eqref{follower} as in the previous model. In particular, if the platoon size $N_m=1$ and $t_d=0$ for all $m$, \eqref{mp-ovm} becomes equivalent to the modified OVM with autonomous vehicles in \cite{zhu2018analysis}.

\begin{remark}
% Extention in control strategies and other CF models and the reasons that these are not discussed in this paper
The proposed frame work can be applied to any general second order CF models and combined with control strategies such as delayed feedback control, CACC, MPC, etc. However, the main focus of this paper is to present a basic framework for multi-platoon CAVs, so we have kept the models as simple as possible with minimal parameters.    
\end{remark}

\section{Stability analysis \label{sec23}}
%no connection (done, need to check if the limit can be calculated) and connected with delay (done follow ngoduy's work, neglect higher order)
In this section, we analyse the stability of the models in section \ref{sec22} through linear stability analysis. The steady-state (equilibrium) solution of all the aforementioned models on a ring road of length $L$ with $N_{tot}=\sum_{i=1}^{m}N_i$ vehicles is
\begin{equation}
    e_{i,j}(t)=h(N_1+\dots+N_{i-1}+j)+V(h)t,
    \label{ch3_equi}
\end{equation}
where $h=L/N_{tot}$ is the equilibrium headway. To analyse the effect of platoon sizes, for the stability analysis we assume that for the multi-platoon models all the platoons are of a uniform size, denoted as $N$. Then for the no-connection model proposed in Subsection \ref{subsec23}, the following stability criterion holds:
\begin{theorem}
    The no-connection multi-platoon model (\ref{leadernocon}, \ref{follower}) with identical platoon size $N$ is stable if
    \begin{equation}
        a>\frac{2NV'(h)}{(N-1)^2+1}.
        \label{crit1}
    \end{equation}
    \label{thm1}
\end{theorem}

\begin{proof}
    Assume that for each vehicle there is a small deviation from the equilibrium solution:
    \begin{equation}
        x_{i,j}(t)=e_{i,j}(t)+y_{i,j}(t),\;|y_{i,j}|\ll 1.
    \end{equation}
    Since the platoon leader is just following the last vehicle of the platoon in front, for the first and last vehicle of each platoon they formulate an sub-system of ODEs of $2m$ equations. And we can linearize the sub-system by doing Taylor expansion of $y_{i,j}$ and neglect higher order terms to get
    \begin{equation}
        \ddot{y}_{i,N}(t)=a\left[V'(h)(y_{i+1,1}(t)-y_{i,N}(t))-\dot{y}_{i, N}(t)\right]
        \label{lineart}
    \end{equation}
    for platoon leaders, and
    \begin{equation}
        \ddot{y}_{i,1}(t)=a\left[V'(h)\frac{y_{i,N}(t)-y_{i,1}(t)}{N-1}-\dot{y}_{i,1}(t)\right]
        \label{linearh}
    \end{equation}
    for the platoon tails, where $i$ is referring to all the integers that satisfies $1\leq i\leq m$ throughout the proof. The $(m+1)$-th platoon is the same as the $1$st platoon. Then if $\lambda$ is an eigenvalue of the linear ODE system, and $\xi_{i,j}$ are the corresponding coefficients of $y_{i,j}$, simplified from (\ref{linearh}, \ref{lineart}) we have
    \begin{equation}
        \lambda^2+a\lambda-aV'(h)\left(\frac{\xi_{i+1,1}}{\xi_{i,N}}-1\right)=0,
        \label{lambda1}
    \end{equation}    
    and
    \begin{equation}
        \lambda^2+a\lambda-aV'(h)\left(\frac{\xi_{i,N}}{(N-1)\xi_{i,1}}-\frac{1}{N-1}\right)=0.
        \label{lambda2}
    \end{equation}
        Then with the same $\lambda$, the constant parts of (\ref{lambda1}) and (\ref{lambda2}) are identical, and we can denote it by $r$. Then the real parts of $\lambda$ can be rewritten as 
    \begin{equation}
        \text{Re}(\lambda)=\frac{1}{2\left(-a+\sqrt{(\frac{d+\sqrt{d^2+e^2}}{2})}\right)},
    \end{equation}
    where $d=a^2+4\text{Re}(r)$ and $e=4\text{Im}(r)$. The sub-system (\ref{linearh}, \ref{lineart}) is stable if $\text{Re}(\lambda)<0$, which can be simplified to 
    \begin{equation}
        a>\left|\frac{\text{Im}^2(r)}{\text{Re}(r)}V'(h)\right|.
        \label{critim}
    \end{equation}

    Now it remains to show \eqref{crit1} implies \eqref{critim}. Note that 
    \begin{equation}
        \prod_{i=1}^m\frac{\xi_{i+1,1}}{\xi_{i,N}}\frac{\xi_{i,N}}{\xi_{i,1}}=1
    \end{equation}
    holds since $m+1=1$ on the ring road, combining with (\ref{lambda1}, \ref{lambda2}) we have $r$ satisfies
    \begin{equation}
        ((N-1)r+1)^m(r+1)^m=1.
        \label{eqnforr}
    \end{equation}
    Then we can solve for $r$ to get
    \begin{equation}
    r=\frac{-N\pm\sqrt{N^2-4(N-1)(1-\exp(\frac{2\pi ki}{m}))}}{2(N-1)},
    \end{equation}
    where $k=1,2,\dots,m$. Let $\theta=\frac{2\pi k}{m}$ and $l=1/(N-1)$, then
    \begin{equation}
    -\frac{\text{Im}^2(r)}{\text{Re}(r)}=\frac{\sqrt{d_2^2+e_2^2}-d_2}{l+1\pm\sqrt{\frac{\sqrt{d_2^2+e_2^2}+d_2}{2}}},
    \label{stabconst}
    \end{equation}
where $d_2=(l+1)^2-4l(1-\cos\theta)$ and $e_2=4l\sin\theta$. Therefore \eqref{stabconst} can be considered as a function of $\theta$, and for $l<1$ this is a decreasing function for $\theta>0$. Then we can obtain
\begin{equation}
    \left|\frac{\text{Im}^2(r)}{\text{Re}(r)}\right|<\lim_{\theta\to0^+}\left(-\frac{\text{Im}^2(r)}{\text{Re}(r)}\right)=\frac{2N}{(N-1)^2+1}.
    \label{upper}
\end{equation}
    Combined \eqref{upper} with \eqref{critim}, the sub-system (\ref{linearh}, \ref{lineart}) is stable if the stability criterion \eqref{crit1} holds. For the remaining vehicles inside each platoon, the solution is only determined by the leading vehicle of the platoon. And by linearization of \eqref{follower} we can rewrite $y_{i,j}$ as
    \begin{equation}
        y_{i,j}=\left(N-j+\frac{j-1}{N-1}\right)y_{i,N},
    \end{equation}
    which is a linear function of $y_{i,N}$. This means the stability of the multi-platoon system is the same as the sub-system (\ref{linearh}, \ref{lineart}). Therefore stability criterion \eqref{crit1} holds.
\end{proof}
% Plot stability lines with respect to platoon size

Figure \ref{stabplot1} is the plot of stability regions with different platoon size of the no-connection model. 
\begin{remark}
    For all the stability plots, each neutral stability line separates the graph into two regions: the region above the line is stable and the region below the line is unstable.
\end{remark}
\begin{figure}[ht]
    \centering
    \includegraphics[width=0.7\linewidth]{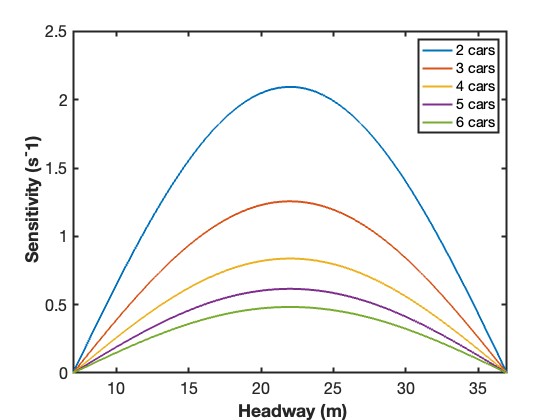}
    \caption{Neutral stability lines of the multi-platoon model with no connection of platoon size $N=2,3,4,5,6$.}
    \label{stabplot1}
\end{figure}

Using similar approaches, for the two-way connected model proposed in Subsection \ref{subsec24}, the following stability criterion applies:
\begin{theorem}
    The two-way connected multi-platoon model (\ref{mp-ovm}, \ref{follower}) with identical platoon size $N$ and constant communication delay $t_d=t_{df}=t_{db}$ is stable if
    \begin{equation}
        a>\frac{2V'(h)}{(1+2p)(N-2t_d V'(h))}.
    \end{equation}
    \label{thm2}
\end{theorem}
\begin{proof}
    For the connected model (\ref{mp-ovm}, \ref{follower}) we follow the assumptions of previous works, e.g. \cite{zhou2014nonlinear, ngoduy2015linear} such that higher orders of the constant delay $t_d$ are neglected. And after linearization, for the connected multi-platoon system in Subsection \ref{subsec24}, The $m$ platoon leaders form a system of $m$ linear ODEs:
\begin{equation}
\begin{aligned}
    \ddot{y}_{i,N}(t)=a&\left[V'(h)(1+p)\frac{\Delta_{i,N}(t-t_d)}{N}-V'(h)p\frac{\Delta_{i-1,N}(t-t_d)}{N}-\dot{y}_{i,N}(t)\right].
    \end{aligned}
\end{equation}
Then if $\lambda$ is an eigenvalue of the system, by reserving first order of $t_d$ via Taylor expansion, we have
\begin{equation}
    \lambda^2+a(1+2p)\lambda-(1-\lambda t_d)\frac{aV'(h)}{N}(1-e^{i\theta})=0,
\end{equation}
where $\theta$ is the same as in the proof of Theorem \ref{thm1}. Then by simplifying the condition $\text{Re}(\lambda)<0$, the stability criterion is equivalent to
\begin{equation}
    4k(1-\cos\theta)+8at_d k^2(1-\cos\theta)^2>k^2\sin^2\theta,
\end{equation}
where $k=a/(N\cdot V'(h))$. Then let $\theta\to 0$ we can get the stability criterion given in Theorem \ref{thm2}.
\end{proof}
Figure \ref{stabplot2} is the plot of the front connected model's stability regions ($p=0,t_d=0$).
\begin{figure}[ht]
    \centering
    \includegraphics[width=0.7\linewidth]{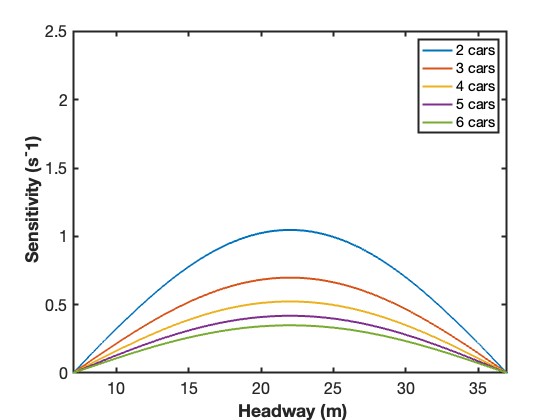}
    \caption{Neutral stability lines of the multi-platoon model with front connection of platoon size $N=2,3,4,5,6$ and zero-delay.}
    \label{stabplot2}
\end{figure}
We observe that the connected model exhibits larger stability regions than the non-connected model. However, the difference diminishes as platoon size increases. Figure \ref{stabplot3} is the plot of stability regions of the two-way connected model with different delays of backward sensitivity $p=0.3$ and platoon size $N=4$. 
\begin{figure}[ht]
    \centering
    \includegraphics[width=0.7\linewidth]{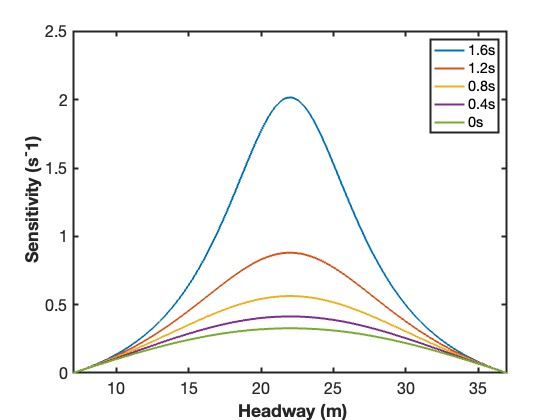}
    \caption{Neutral stability lines of the multi-platoon model with two-way connection of platoon size $N=4$ and delay $t_d=0,0.35,0.8,1.2,1.6$s.}
    \label{stabplot3}
\end{figure}
From this figure, we can observe that the effect of delay become larger as it get close to $1.6$s. Additionally, if the delay reach $2$s then the model becomes consistently unstable for headways between $20$ and $25$ meters. 
\begin{remark}
    For the effect of backward sensitivity $p$ we refer to \cite{nakayama2001effect, zhu2018analysis}.
\end{remark}

% Plot stability lines with respect to delay time and fixed platoon size 2,4. State that too delayed information can cause instability regardless of parameters and some discoveries are consistent with previous papers on this topic
\section{Numerical simulations \label{sec24}}

\subsection{General information} \label{subsec41}
We use MATLAB for both simulation and plots throughout this section. The simulations are performed on a single-lane roads. To acquire more realistic results, we modified the OVM by adding a maximum acceleration constraint and an emergency braking system as follows:
\begin{itemize}
    \item Maximum acceleration constraint: Due to mechanical limits, the maximum allowable acceleration can be less than the theoretical value predicted by the optimal velocity model. For the simulations in this section, we add a constant maximum acceleration constraint, denoted as $a_m$.
    \item Emergency braking system: To avoid collisions, we implement an emergency braking system. If the headway of two adjacent vehicles get smaller than the safety headway $h_m$ (which is a function of the current speed and relative speed between the two vehicles), then the vehicle behind brakes with emergency braking deceleration $a_b$.
\end{itemize}
The modified OVM for HDVs is then given by
\begin{equation}
    \ddot{\tilde{x}}_i=\begin{cases}
        \min\left(\ddot{x}_i,a_m\right),\text{ if }x_{i}(t)-x_{i+1}(t)\geq h_m;\\
        a_b,\text{ if }x_{i}(t)-x_{i+1}(t)< h_m,
    \end{cases}
    \label{movm}
\end{equation}
% Model parameters
where the acceleration term $\ddot{x}_i$ is given by \eqref{ch3_OVM1}.We modify the multi-platoon models accordingly by substituting $\ddot{x}_i$ with other acceleration functions. The solutions of the models are obtained in discrete forms using a modified Euler scheme: 
\begin{equation}
    \begin{cases}
    \dot{x}_{i,j+1}=\dot{x}_{i,j}+\ddot{\tilde{x}}_{i,j}\Delta t;\;\\
    x_{i,j+1}=x_{i,j}+\frac{\dot{x}_{i,j}+\dot{x}_{i,j+1}}{2}\Delta t,
    \end{cases}
\end{equation}
where $\Delta t$ is the uniform time step size, $x_{i, j}$, $\dot{x}_{i,j}$, $\ddot{\tilde{x}}_{i,j}$ are the position, velocity, acceleration of the $i$-th car at the $j$-th time step of simulation, respectively. This scheme is equivalent to the ones in \cite{zhu2018analysis} and \cite{hui2024new}. We also use a consistent time step size of $\Delta t=0.1$ seconds. 

The model parameters are set as follows: The total length of the ring road is $L=2640$m, with a total of $N_{tot}=120$ vehicles. All simulations run for the same duration $T=4000$ seconds. We set the maximum acceleration constraint to $a_m=3$m/s$^2$, the emergency braking deceleration to $a_b=-8$m/s$^2$ and define the safety headway as
\begin{equation}
    h_m(v_i,v_{i+1})=\frac{(v_{i}-v_{i+1})^2}{|2a_b|}+\tau(v_{i}-v_{i+1})+l,
\end{equation}
where $v_i$ is the speed of the $i$-th vehicle, $\tau=4$ is the constant time headway for safety, and $l=5$m is the length of each vehicle. We fix the forward sensitivity as $a=0.6$ and backward sensitivity as $p=0.3$ if available. The optimal velocity function is given by equation \eqref{ch3_ovf1}, with parameters $h_{\min}=7$m, $h_{\max}=37$m, $v_{\max}=20$m/s. The equilibrium headway and velocity are calculated as $h=L/N=22$m and $V(h)=10$ m/s, respectively. The initial position and velocity of the $i$-th vehicle are deviated from the equilibrium states $(e_i,V(h))$ with random perturbations uniformly distributed on the interval $[-5/2,5/2]$. The initial condition of the model is given by
\begin{equation}
    \begin{cases}
        x_i(0) =& e_i(0)+r_i, \\
        \dot{x}_i(0) = & V(h)+\bar{r}_i
    \end{cases}
\end{equation}
where $r_i,\bar{r}_i$ are random values generated from a uniform distribution over $[-5/2,5/2]$, and $e_i(0)=hi$ can be calculated from equation \eqref{ch3_equi}. In the following subsections, we introduce three sets of simulations involving CAV platoons of varying platoon sizes, communication levels, and distributions in mixed traffic.
\begin{remark}
    For readers interested in variations of sensitivity parameters, studies including \cite{bando1995dynamical}, \cite{lenz1999multi}, \cite{zhu2018analysis}, \cite{hui2024new} explore different sensitivity parameter settings in various simulations.
\end{remark}

\subsection{Experiments of platooned CAVs}
% Goal: platoon size, how the simulations are performed: no-delay no-backward detection model for platoon 
% size=2,3,4,5, figures and introduction of figures: probably need snapshot for velocity, analyse simulation results.
In this subsection, we conduct experiments on traffic flow consisting solely of identically-sized CAV platoons, aiming to investigate the effects of various factors such as platoon size, connectivity, and communication delay.

\subsubsection{Different platoon sizes}
% Goal: check for communication, how the simulations are performed: delayed no-backward detection, delayed with backward detection with different numbers, model for platoon 
% size=4. figures and introduction of figures: 2 sets of figures, probably need hysterisis loops for velocity, analyse simulation results
In this simulation, we aim to show the effect of platoon size and connectivity without delay. The platoon sizes are selected as $N=2,3,4,5$, and the connectivity options between platoons include no-connection, front-connection, and two-way connection. Figure \ref{sim1-1-1h} is the headway plots for the no-connection model with platoon size $N=2,3,4,5$ (it is not necessary to explore $N>5$ since $N=5$ is already stabilized).
\begin{figure}[ht]
    \centering
    \subfloat[$N=2$]{\includegraphics[width=0.35\textwidth]{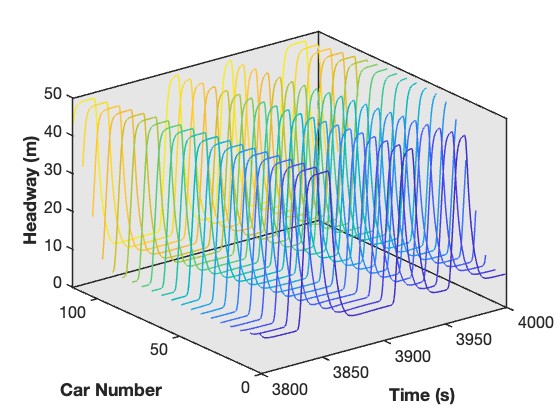}}
    \subfloat[$N=3$]{\includegraphics[width=0.35\textwidth]{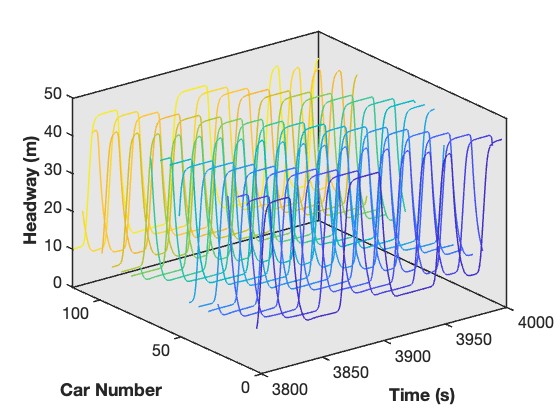}}
    \\
    \subfloat[$N=4$]{\includegraphics[width=0.35\textwidth]{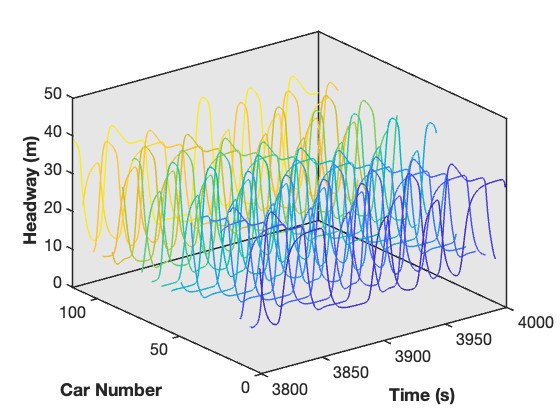}}
    \subfloat[$N=5$]{\includegraphics[width=0.35\textwidth]{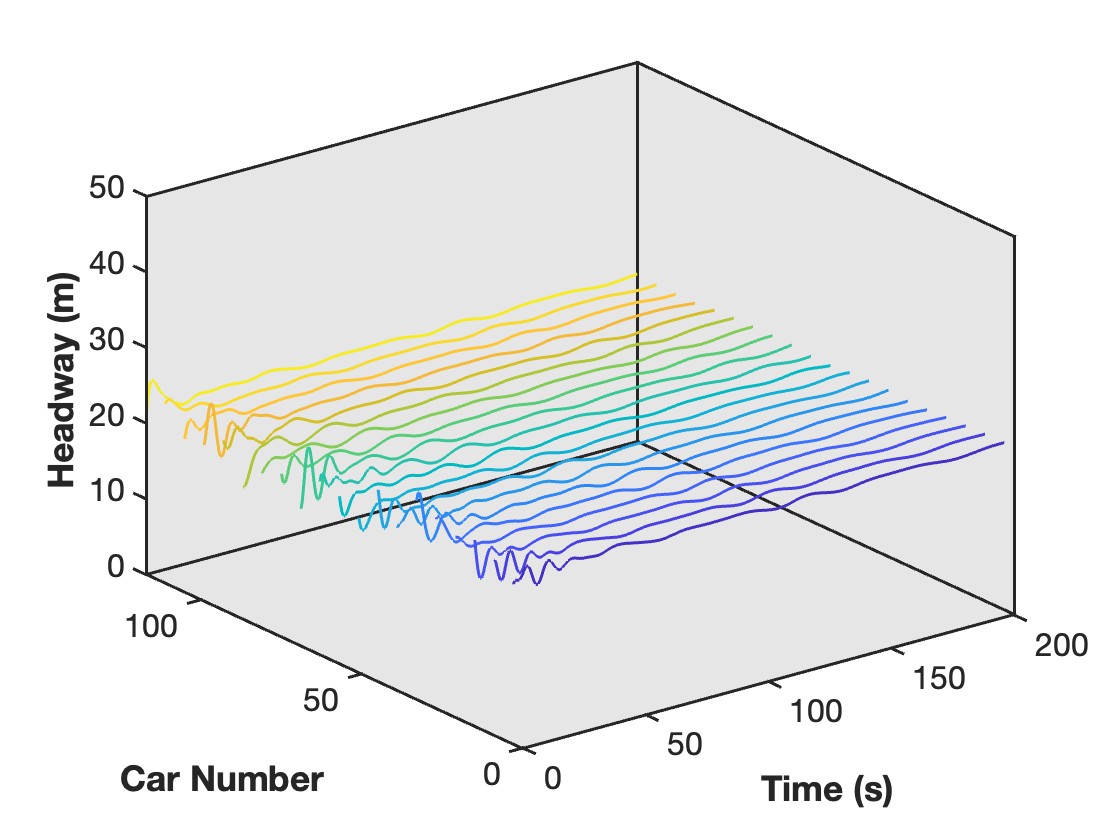}}
    \caption{Plots of headways for selected vehicles with no-connection for platoon sizes $N=2,3,4,5$.}
    \label{sim1-1-1h}
\end{figure}

\begin{remark}
     For the headway plots in this section, if the vehicles have not stabilized after $4000$s, we select every 6th vehicle from the 1st to the 120th for plotting (20 vehicles in total), with the time interval spanning from 3800 to 4000 seconds. If stabilized before $4000$s, the headway plots will instead cover the time period from $0$s to stabilization ($60$ to $300$s, depending on the scenarios). 
\end{remark}
Figure \ref{sim1-1-1v} is the plots of minimum and maximum speeds of all vehicles corresponding to Figure \ref{sim1-1-1h}.
\begin{figure}[ht]
    \centering
    \subfloat[$N=2$]{\includegraphics[width=0.35\textwidth]{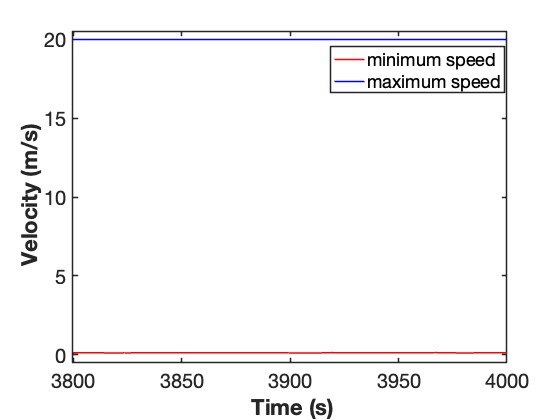}}
    \subfloat[$N=3$]{\includegraphics[width=0.35\textwidth]{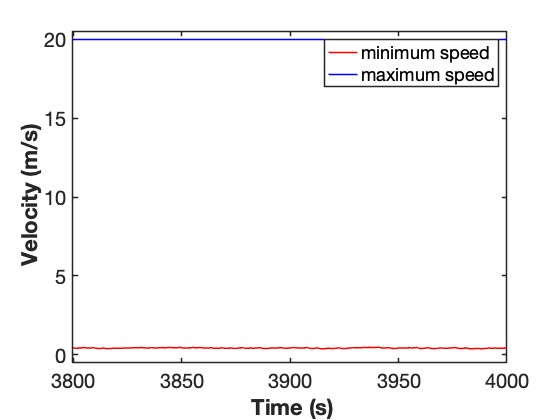}}
    \\
    \subfloat[$N=4$]{\includegraphics[width=0.35\textwidth]{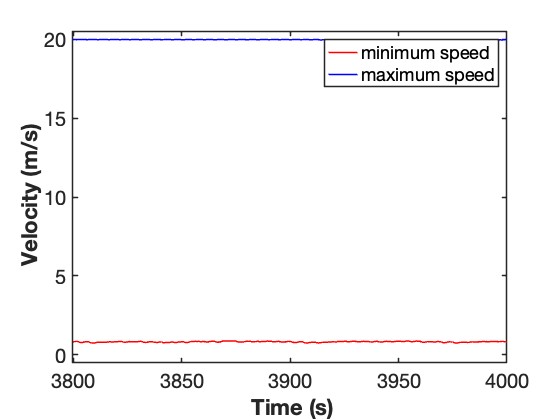}}
    \subfloat[$N=5$]{\includegraphics[width=0.35\textwidth]{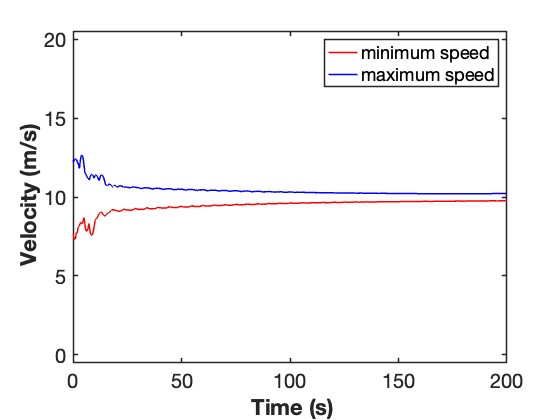}}
    \caption{Plots of minimum and maximum speeds of platoons with no-connection and size $N=2,3,4,5$.}
    \label{sim1-1-1v}
\end{figure}
Figure \ref{sim1-1-2h} is the headway plots for front-connected platoons of size $N=2,3,4$ and for two-way connected platoons of size $N=2$.
\begin{figure}[ht]
    \centering
    \subfloat[$N=2$ front]{\includegraphics[width=0.35\textwidth]{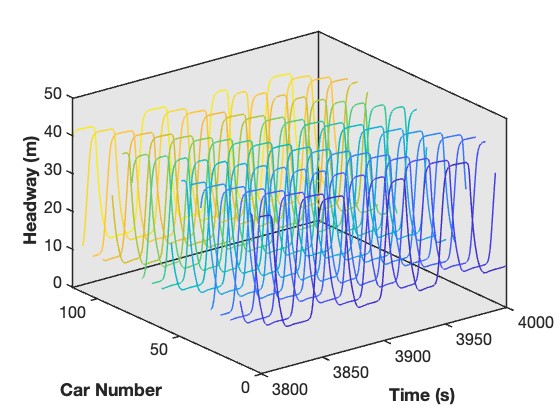}}   
    \subfloat[$N=3$ front]{\includegraphics[width=0.35\textwidth]{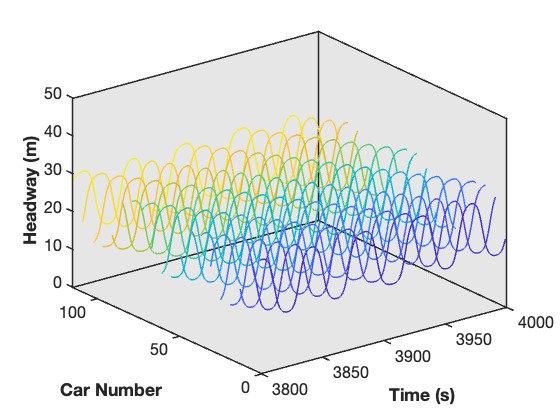}}
    \\
    \subfloat[$N=4$ front]{\includegraphics[width=0.35\textwidth]{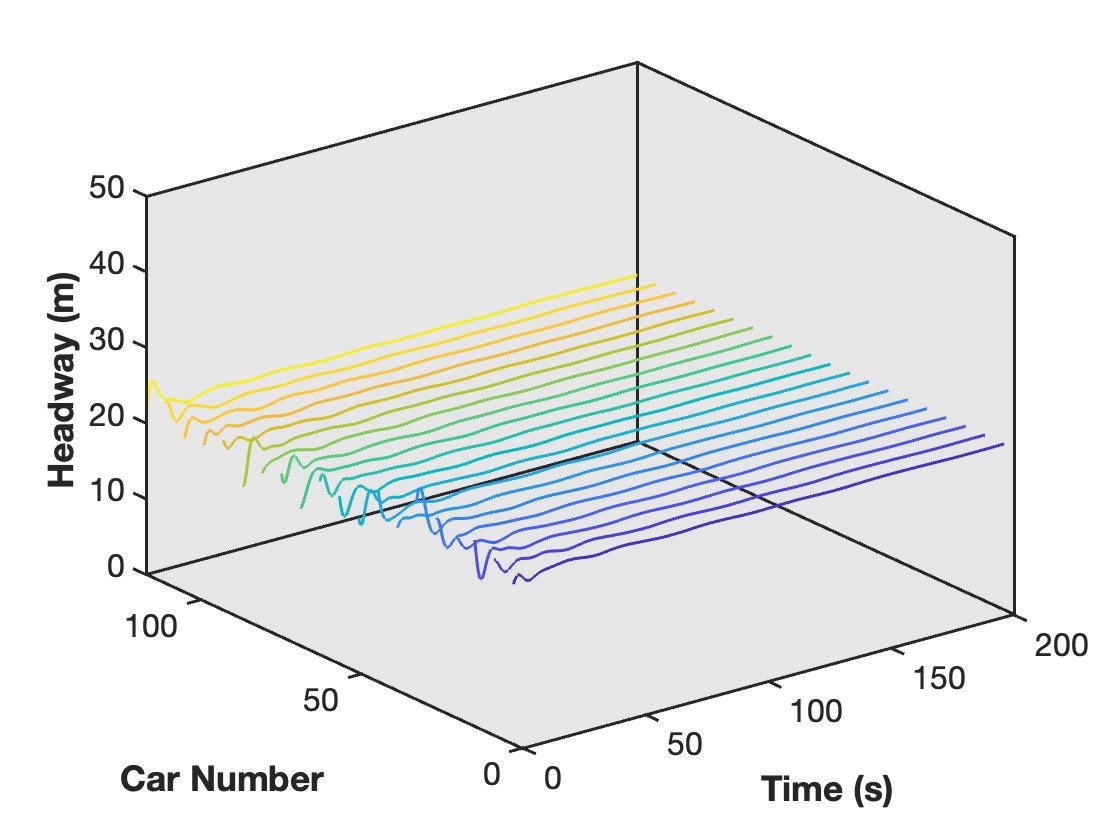}}
    \subfloat[$N=2$ two-way]{\includegraphics[width=0.35\textwidth]{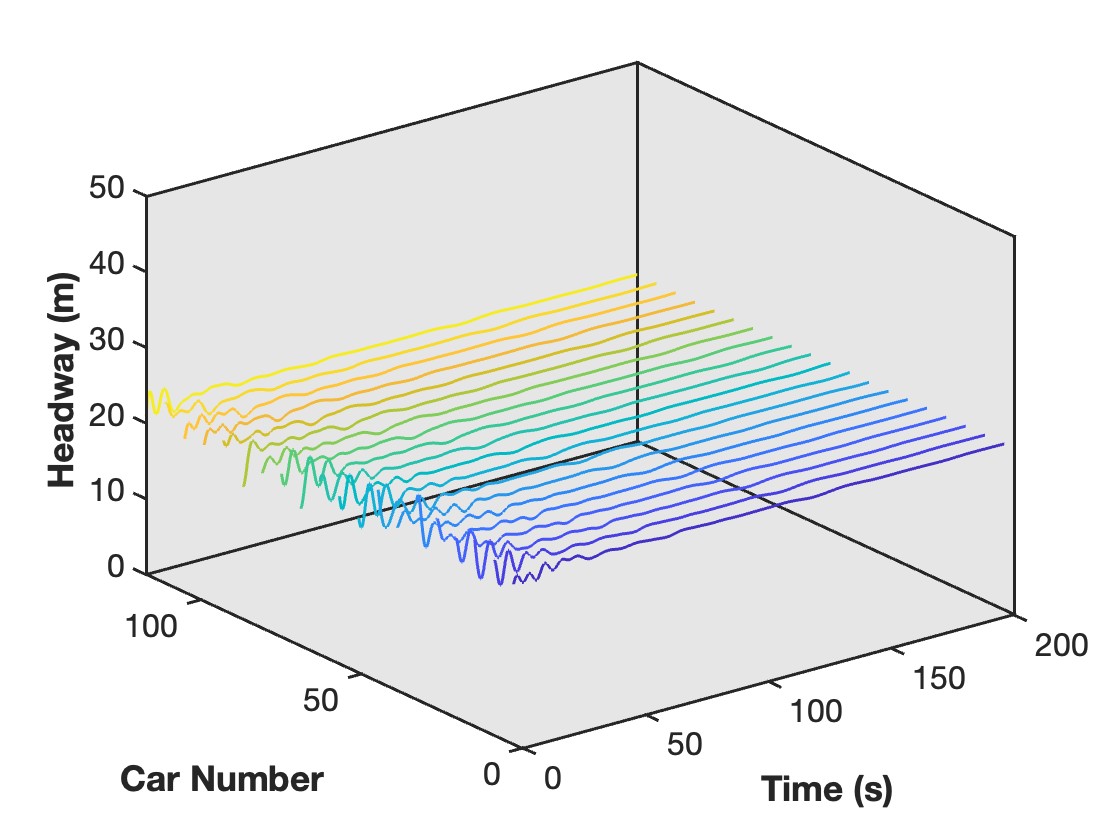}}
    \caption{Plots of headways for selected vehicles with front connection for platoon sizes $N=2,3,4$, two-way connection for platoon size $N=2$.}
    \label{sim1-1-2h}
\end{figure}

From the simulation results, we observe that the stability of CAV platoons improves by increasing platoon size when intra-platoon communication is robust. Moreover, with front and two-way connections, the equilibrium state can be achieved with platoons consisting of just four and two CAVs, respectively. However, this condition is satisfied only if the inter-platoon communication is delay-free. These findings highlight the critical role of intra-platoon V2V communication in maintaining stability, suggesting that prioritizing strong intra-platoon connectivity is more beneficial than relying on inter-platoon communication to achieve stability.

\subsubsection{Fixed platoon size with delays}
In this simulation, our objective is to find the effects of inter-platoon communication delay. We fix the platoon sizes at $N=4$ and apply constant and stochastic delays for the two-way connected model. Figure \ref{sim1-2-1h} is the headway plots for the two-way connected model with communication delays $t_d=0.35, 0.8, 1.2, 1.6$s. 
\begin{figure}[ht]
    \centering
    \subfloat[$t_d=0.35$s]{\includegraphics[width=0.35\textwidth]{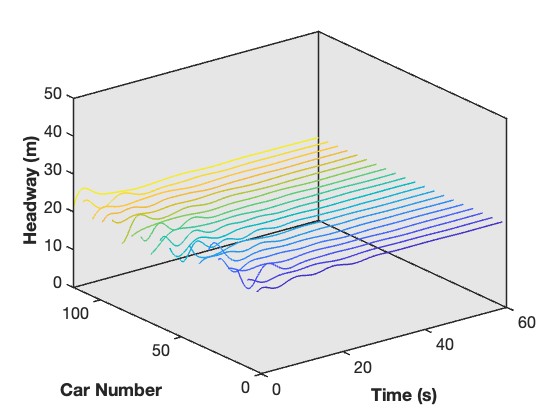}}    
    \subfloat[$t_d=0.8$s]{\includegraphics[width=0.35\textwidth]{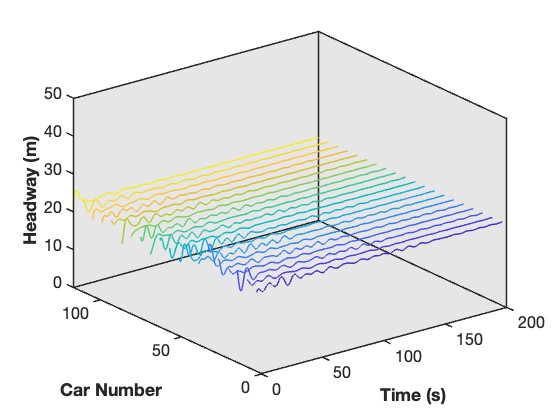}}
    \\
    \subfloat[$t_d=1.2$s]{\includegraphics[width=0.35\textwidth]{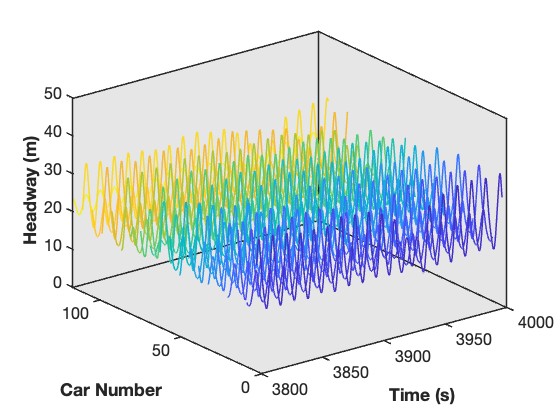}}    
    \subfloat[$t_d=1.6$s]{\includegraphics[width=0.35\textwidth]{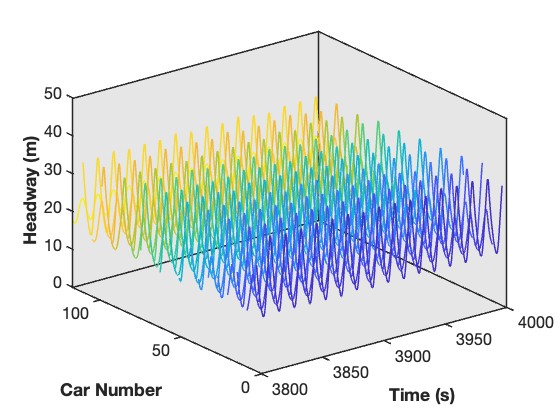}}
    \caption{Plots of headways for selected vehicles with two-way connection for platoon size $N=4$ and communication delays $t_d=0.35, 0.8, 1.2, 1.6$s.}
    \label{sim1-2-1h}
\end{figure}
Figure \ref{sim1-2-1v} is the plots of minimum and maximum speeds of all vehicles corresponding to Figure \ref{sim1-2-1h}.
\begin{figure}[ht]
    \centering
    \subfloat[$t_d=0.35$s]{\includegraphics[width=0.35\textwidth]{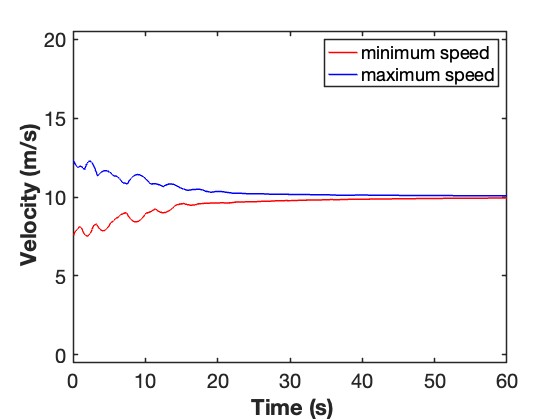}}   
    \subfloat[$t_d=0.8$s]{\includegraphics[width=0.35\textwidth]{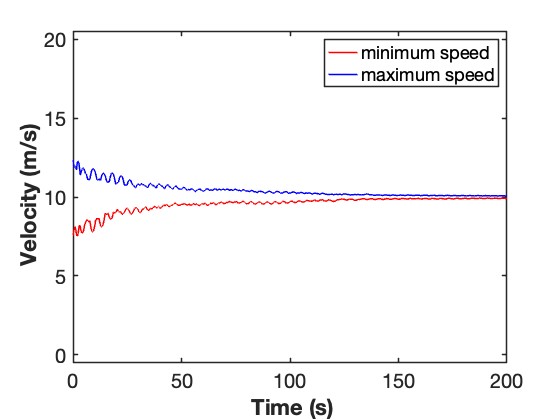}}
    \\
    \subfloat[$t_d=1.2$s]{\includegraphics[width=0.35\textwidth]{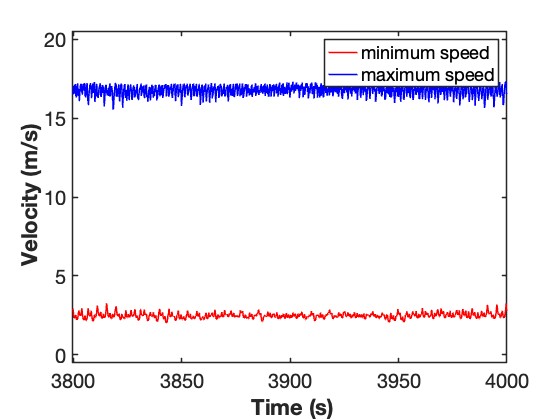}}    
    \subfloat[$t_d=1.6$s]{\includegraphics[width=0.35\textwidth]{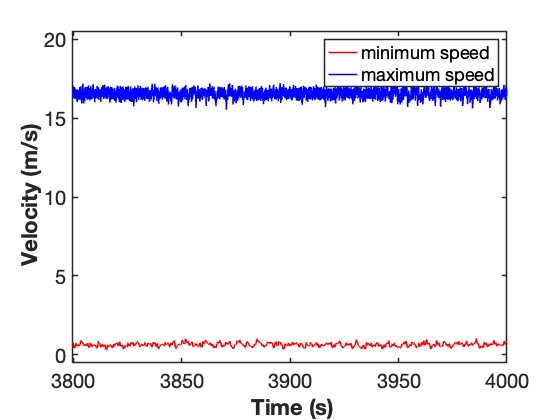}}
    \caption{Plots of minimum and maximum speeds of platoons with two-way connection of size $N=4$ and communication delays $t_d=0.35, 0.8, 1.2, 1.6$s.}
    \label{sim1-2-1v}
\end{figure}
Figure \ref{sim1-2-2} is the plots of headways and extreme speeds of all vehicles corresponding to Figure \ref{sim1-2-1h}.
\begin{figure}[ht]
    \centering
    \subfloat[Headway: $t_d=0.6 + r_d$s]{\includegraphics[width=0.35\textwidth]{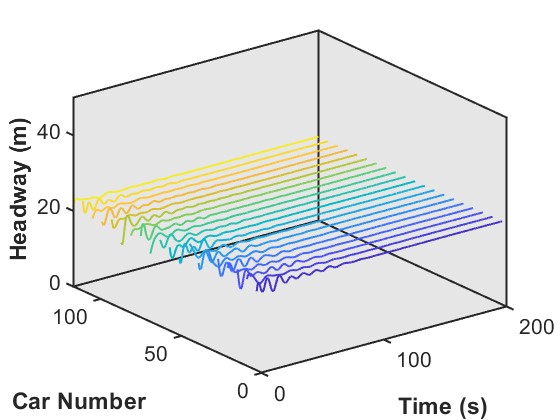}}
    \subfloat[Headway: $t_d=1 + r_d$s]{\includegraphics[width=0.35\textwidth]{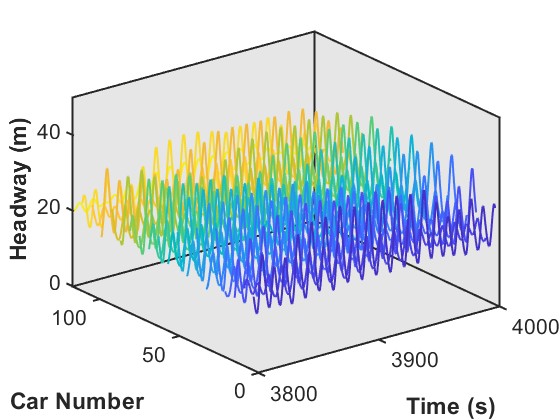}}
    \\
    \subfloat[Speed: $t_d=0.6 + r_d$s]{\includegraphics[width=0.35\textwidth]{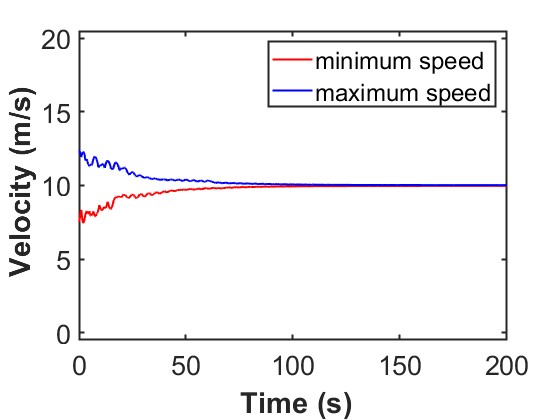}}
    \subfloat[Speed: $t_d=1 + r_d$s]{\includegraphics[width=0.35\textwidth]{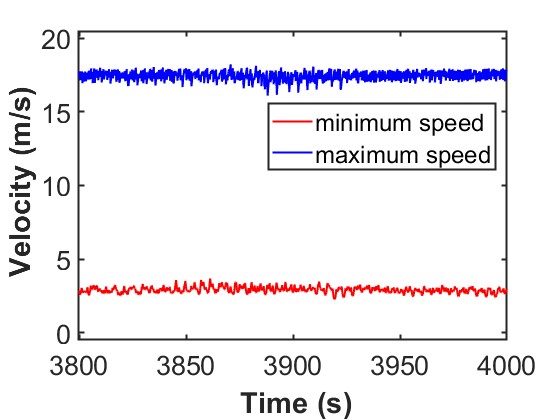}}
    \caption{Plots of headways and extreme speeds of platoons with two-way connection of size $N=4$ and communication delays $t_d = 0.6 + r_d, 1 + r_d$s.}
    \label{sim1-2-2}
\end{figure}

From the simulation results, we observe that increasing communication delays between platoons negatively impacts the stability of CAV platoons. Moreover, the variance in headways grows exponentially with increased delay, which aligns with theoretical analysis. The results from random delay shows that the noises from communication have negligible to marginal impacts under certain threshold. These findings can guide CAV manufacturers in setting standards for sensors and other hardware components that influence communication delays.

\subsection{Experiments of mixed autonomy}
% Goal: check for mixed flow, how the simulations are performed: no-delay, no-connection for platoon size=5, =8, with concentrated distribution and separated distribution, 4 set of figures, each with 2.
% figures and introduction of figures: probably need snapshot for velocity, analyse simulation results
% Closer to the tail of platoons gets more benefits (less variation). Which is due to string instability of HDVs
One potential benefit of implementing CAV platoons in traffic flow is their stabilizing effect when mixed with HDVs (which can be treated as CAV platoons of size 1 with no communication). In this subsection, we test various distributions of CAV platoons and HDVs on the ring road described in Subsection \ref{subsec41} to evaluate their impact on traffic flow stability.

\subsubsection{Segregated CAV platoons and HDVs}
We first consider the scenario where CAV platoons and HDVs are segregated into two distinct groups, each forming its own string of vehicles. The platoon configurations are set with sizes of $N=6$ or $N=8$, with no inter-platoon connections. Figure \ref{sim2-1-1h} is the headway plots for segregated traffic with CAV platoons of size $N=6$ with either $24$ or $30$ HDVs, and CAV platoons of size $N=8$ with either $32$ or $40$ HDVs.
\begin{figure}[ht]
    \centering
    \subfloat[$N=6$; $30$ HDVs]{\includegraphics[width=0.35\textwidth]{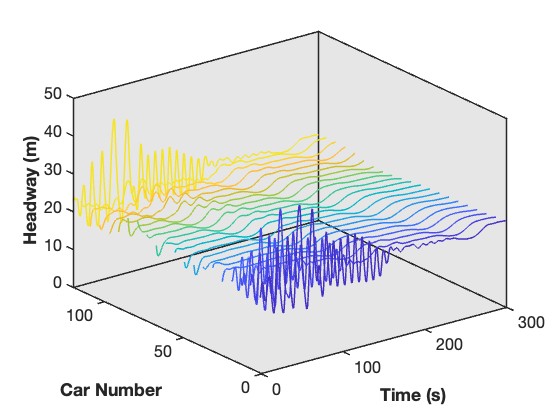}}
    \subfloat[$N=6$; $36$ HDVs]{\includegraphics[width=0.35\textwidth]{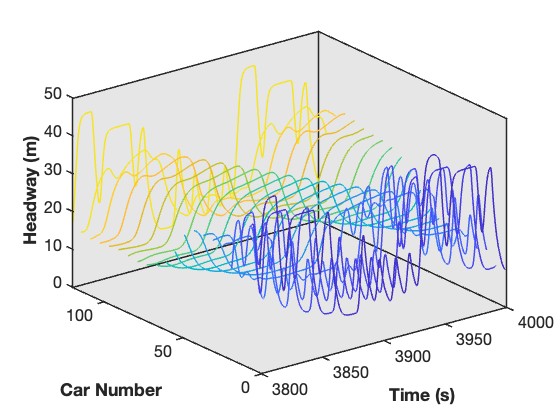}}
    \\
    \subfloat[$N=8$; $32$ HDVs]{\includegraphics[width=0.35\textwidth]{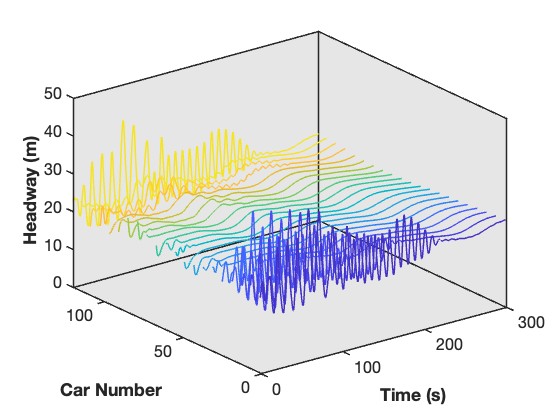}}
    \subfloat[$N=8$; $40$ HDVs]{\includegraphics[width=0.35\textwidth]{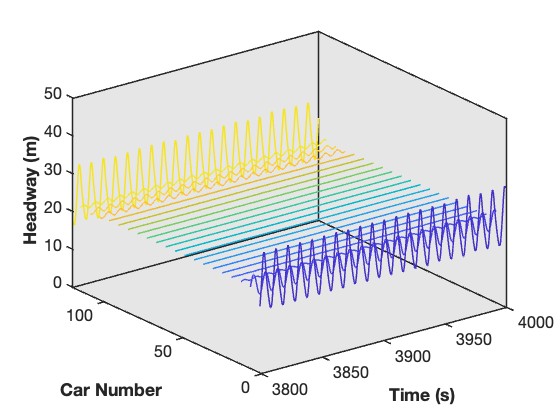}}
    \caption{Headway plots for segregated traffic with CAV platoons of size $N=6$ with $30$ and $36$ HDVs and CAV platoons of size $N=8$ with $32$ and $40$ HDVs.}
    \label{sim2-1-1h}
\end{figure}

From this simulation, we observe that platoons of size $N=6$ can stabilize up to $30$ HDVs, slightly less than $32$ HDVs stabilized by platoons of size $N=8$.(It is worth noting that with 40 HDVs, the traffic flow is nearly stable.) Moreover, if the model does not reach equilibrium, the speed variation of HDVs is smaller when they are positioned closer to the tail CAV of the platoons, suggesting that HDVs are also prone to string instability.

\subsubsection{Evenly mixed CAV platoons and HDVs}
Another appraoch to distributing the vehicles is to mix CAV platoons and HDVs as evenly as possible, i.e. each platoon is follow by a fixed number of HDVs. We again assume that the CAV platoons are not connected, as the distance between platoon leaders is longer than flows of only CAV platoons, which could result in increased delays. Figure \ref{sim2-2-1h} is the headway plots for evenly mixed traffic, where CAV platoons of size $N=6$ are followed by $2$ or $3$ HDVs, and CAV platoons of size $N=8$ are followed by $5$ or $6$ HDVs.
\begin{remark}
    One of the platoons may be followed by a different number of HDVs to balance the distribution.
\end{remark}
\begin{figure}[ht]
    \centering
    \subfloat[$N=6$; $2$ HDV followers]{\includegraphics[width=0.35\textwidth]{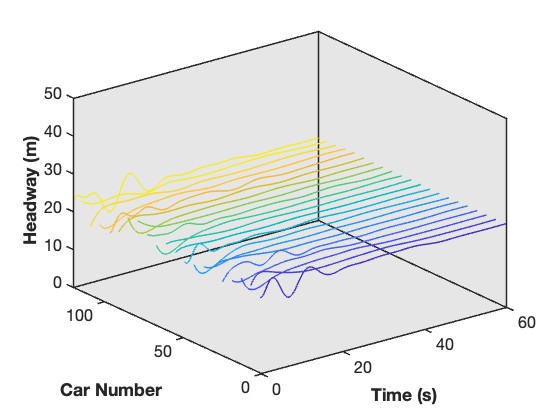}}
    \subfloat[$N=6$; $3$ HDV followers]{\includegraphics[width=0.35\textwidth]{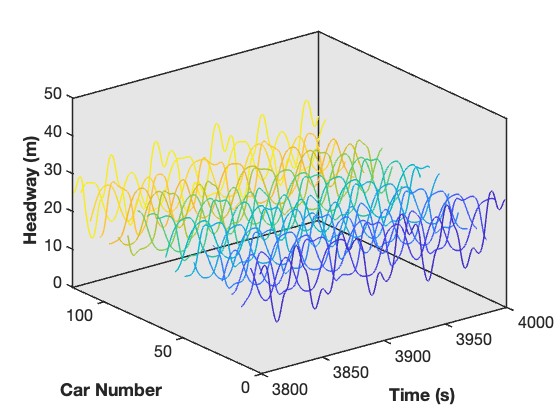}}
    \\
    \subfloat[$N=8$; $5$ HDV followers]{\includegraphics[width=0.35\textwidth]{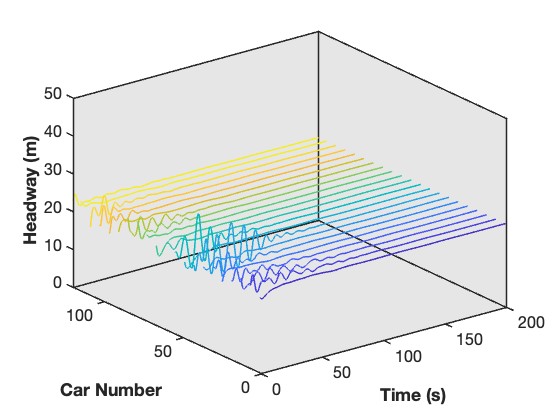}}
    \subfloat[$N=8$; $6$ HDV followers]{\includegraphics[width=0.35\textwidth]{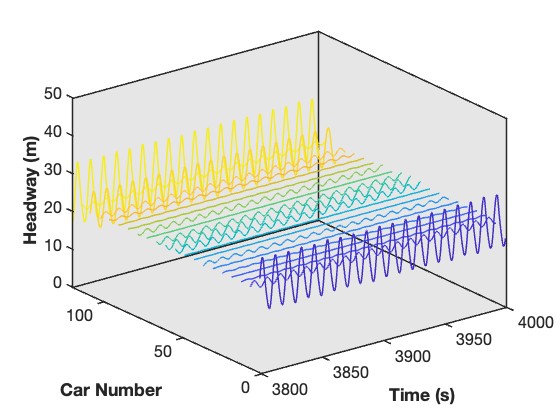}}

    \caption{Headway plots for evenly mixed traffic with each CAV platoon of size $N=6$ followed by $2$ or $3$ HDVs, and CAV platoon of size $N=8$ followed by $5$ or $6$ HDVs.}
    \label{sim2-2-1h}
\end{figure}

From this simulation, we observed that in evenly mixed distributions, platoons of size $N=6$ can stabilize up to $30$ HDVs, while platoons of size $N=8$ can stabilize up to $48$ HDVs. This suggests that larger platoons act as more effective controllers of traffic stability. However, in scenarios with segregated distributions of $N=8$ with $40$ HDVs, the traffic flow is nearly stable, indicating that the improvements provided by the even distribution are relatively minor. where stability is not fully achieved, the speed variation among HDVs decreases more significantly when they are positioned closer to the platoons, consistent with previous observations.
% more conclusions can be added as a summarized conclusion of numerical experiments

\FloatBarrier

\section{Conclusions \label{sec25}}
%Contribution in model part, analysis part and experiment part
In this paper, we have extended a recently proposed single-platoon CF model to accommodate multiple platoons. By prioritizing the leading vehicle of each platoon, we proposed two models that account for varying degrees of connectivity between platoons. We showed that our proposed multi-platoon models are consistent with the foundational CF models when the platoon size is reduced to $1$. 

Through linear stability analysis, we demonstrated that both platoon size and the level of inter-platoon communication can enhance system stability. The results of numerical experiments with varied platoon size and connectivity are consistent with theoretical analysis. Furthermore, when testing configurations that mixed CAV platoons with HDVs, we observed that HDVs benefit from following CAV platoons, even without specific design considerations for HDV control. This reveals a more intrinsic stabilizing effect of structured platooning. 
% (For different arrangements experiments have been performed July 17th)

A notable outcome of our analysis is that the influence of inter-platoon connections diminishes as platoon sizes increase. This occurs because larger platoons act as “self-stabilizing units”, where intra-platoon V2V control is sufficient to maintain stability, reducing the need for external coordination. This suggests that in practical deployment, prioritizing robust intra-platoon V2V communication is more effective than improving inter-platoon V2I links, especially in high-density traffic scenarios. Another finding is that the stability of mixed traffic flow exhibits similar characteristics in scenarios where CAV platoons are evenly distributed or segregated, a result that consists with a study of macroscopic models of mixed flow \cite{hui2024anisotropic}.

This paper provides a solid foundational structure for future innovations in CAV technologies and opens several avenues for further exploration. Integration with other control strategies, such as feedback and optimal control, could significantly enhance stability, safety, and comfort for travelers. Additionally, addressing fairness within the model and considering dynamic leader switching and platoon reformation could lead to more practical and equitable applications. Extending the proposed platoon models to accommodate more complex traffic scenarios, such as multi-lane roads and signalized intersections, would broaden the models' applicability. 

\bibliographystyle{unsrt}
\bibliography{raf}

\end{document}